\documentclass[onefignum,onetabnum]{siamart171218}

\usepackage{lipsum}
\usepackage{amsfonts}
\usepackage{graphicx}
\usepackage{epstopdf}
\usepackage{algorithm}
\usepackage{algorithmic}
\usepackage{booktabs}
\usepackage{multirow}
\usepackage{wrapfig}
\ifpdf
  \DeclareGraphicsExtensions{.eps,.pdf,.png,.jpg}
\else
  \DeclareGraphicsExtensions{.eps}
\fi

\usepackage{amsmath}
\usepackage{amssymb}
\usepackage{siunitx}
\usepackage{dsfont}
\usepackage{xcolor}
\usepackage{bbm}

\usepackage{mathtools}
\newcommand{\R}{\mathbb{R}}

\newcommand{\tr}{\text{tr}}
\newcommand{\sumion}{\ensuremath{\sum_{i=1}^N}}
\newcommand{\sumio}[2][i]{\ensuremath{\sum_{#1=1}^#2}}
\newcommand{\avgion}[1][i]{\ensuremath{\frac{1}{N}\sum_{#1=1}^N}}

\newcommand{\liston}[2][n]{\ensuremath{#2_{1},\dotsc,#2_{#1}}}

\newcommand{\dom}{\mathcal{X}}

\newcommand{\N}{\mathbb{N}}
\newcommand{\ind}[1]{\mathbbm{1}\{#1\}}
\newcommand{\Ex}[2][ ]{\ensuremath{\mathbb{E}_{#1}\left[#2\right]}}
\newcommand{\Var}[2][ ]{\ensuremath{\text{Var}_{#1}\left(#2\right)}}
\renewcommand{\Pr}[2][ ]{\ensuremath{\mathbb{P}_{#1}\left(#2\right)}}

\newcommand{\Norm}[1]{\ensuremath{\mathcal{N}\left(#1\right)}}

\newcommand{\simiid}{\overset{\textrm{i.i.d.}}{\sim}}

\newcommand{\PCord}{\ensuremath{m}}
\newcommand{\Mmat}{\hat{M}_{\PCord,\sigma_0}}
\newcommand{\initialset}{\ensuremath{\mathcal{X}_0} }

\newcommand{\distset}{\ensuremath{\mathcal{D}} }
\newcommand{\rs}{\ensuremath{R_{[t_0,t_1]}}}
\newcommand{\ars}{\ensuremath{\hat{R}_{[t_0,t_1]}}}

\newcommand{\cclass}{\mathcal{C}}
\newcommand{\risk}[1][c]{\ensuremath{r(#1)}}
\newcommand{\erisk}[1][c]{\ensuremath{\hat{r}(#1)}}
\newcommand{\srisk}[1][Q]{\ensuremath{r_{#1}}}
\newcommand{\esrisk}[1][Q]{\ensuremath{\hat{r}_{#1}}}

\usepackage{mathtools}
\DeclarePairedDelimiter{\ceil}{\lceil}{\rceil}

\newtheorem{problem}{Problem}

\newsiamremark{remark}{Remark}
\newsiamremark{hypothesis}{Hypothesis}
\crefname{hypothesis}{Hypothesis}{Hypotheses}
\newsiamthm{claim}{Claim}

\headers{Data-driven Reachability with Christoffel Functions}{A. Devonport, F.Yang, L. El Ghaoui, and M. Arcak}

\title{Data-Driven Reachability analysis and Support set Estimation with Christoffel Functions%
    \thanks{Submitted to the editors December 17, 2021. This paper is a revision and extension of a conference paper~\cite{devonport2021datadriven}.
\funding{This work is funded in part by the Air Force Office of Scientific
Research grant FA9550-21-1-0288, National Science Foundation grant ECCS-1906164,
and the Office of Naval Research grant N00014-18-1-2209.}
}}

\author{Alex Devonport 
    \thanks{University of California, Berkeley,
    Berkeley, CA\\ (\email{\{alex\_devonport,forestyang,elghaoui,arcak\}@berkeley.edu})}
    \and Forest Yang\footnotemark[2]
    \and Laurent El Ghaoui\footnotemark[2] 
    \and Murat Arcak\footnotemark[2]
}

\usepackage{amsopn}

\makeatletter
\newcommand*{\addFileDependency}[1]{%
  \typeout{(#1)}%
  \@addtofilelist{#1}%
  \IfFileExists{#1}{}{\typeout{No file #1.}}%
}
\makeatother

\ifpdf
\hypersetup{
  pdftitle={Data-Driven Reachability and Support Estimation with Christoffel Functions},
  pdfauthor={A. Devonport, F. Yang, L. El Ghaoui, and M. Arcak}
}
\fi

\begin{document}

\maketitle

\begin{abstract}
We present algorithms for estimating the forward reachable set of a dynamical
system using only a finite collection of independent and identically distributed
samples. The produced estimate is the sublevel set of a function called an
empirical inverse Christoffel function: empirical inverse Christoffel functions
are known to provide good approximations to the support of probability
distributions. In addition to reachability analysis, the same approach can be
applied to general problems of estimating the support of a random variable,
which has applications in data science towards detection of novelties and
outliers in data sets. In applications where safety is a concern, having a
guarantee of accuracy that holds on finite data sets is critical. In this paper,
we prove such bounds for our algorithms under the Probably Approximately Correct
(PAC) framework. In addition to applying classical Vapnik-Chervonenkis (VC)
dimension bound arguments, we apply the PAC-Bayes theorem by leveraging a formal
connection between kernelized empirical inverse Christoffel functions and
Gaussian process regression models. The bound based on PAC-Bayes applies to a
more general class of Christoffel functions than the VC dimension argument, and
achieves greater sample efficiency in experiments.
\end{abstract}



\section{Introduction}

Reachability analysis is a popular and effective way to guarantee the safety of a system
in the face of uncertainty. The primary object of study is
the reachable set, which characterizes all possible evolutions of a system under
certain constraints on initial conditions and disturbances.
Many algorithms in reachability analysis use detailed
system information to compute a sound approximation to the reachable set, that
is an approximation guaranteed to completely contain (or be contained in) the
reachable set. However,
in many important applications, such as 
complex
cyber-physical systems that
are only accessible through simulations or experiments, this detailed
system information is not available, so these algorithms cannot be applied.
Applications such as these motivate \emph{data-driven} reachability analysis,
which studies algorithms to estimate reachable sets using the type of data that
can be obtained from experiments and simulations. These algorithms have the
advantage of being able to estimate the reachable sets of any system whose
behavior can be simulated or measured experimentally, without requiring any
additional mathematical information about the system. The main disadvantage of
data-driven reachability algorithms is that generally they cannot provide the
same type of soundness guarantees as traditional reachability analysis
algorithms; however, they can still guarantee accuracy of the estimates in a
probabilistic sense with high confidence, as this article will show.

Data-driven reachability is a rapidly growing area of research within
reachability analysis. Many recent developments focus either on providing probabilistic
guarantees of correctness for data-driven methods that estimate the reachable
set directly from data,
for instance using results from statistical learning theory~\cite{devonport2020data} or scenario
optimization~\cite{marseglia2014hybrid,yang2016multi,ioli2017smart,sartipizadeh2019voronoi,hewing2019scenario,devonport2020estimating}. 
Others incorporate data-driven elements into more traditional reachability
approaches, for instance estimating entities such as discrepancy
functions~\cite{fan2017dryvr} or
differential inclusions~\cite{djeumou2020fly}. Further developments include incorporating
data-driven reachability into verification tools for
cyber-physical systems~\cite{fan2017dryvr, qi2018dryvr}.

This paper investigates a data-driven reachability algorithm that directly
estimates the reachable set from data using the sublevel sets of an empirical
inverse Christoffel function, and provides a probabilistic guarantee of accuracy
for the method using statistical learning-theoretic methods. Christoffel
functions are a class of polynomials defined with respect to measures on $\R^n$:
a single measure defines a family of Christoffel function polynomials.
When the measure in question is defined by a probability distribution on $\R^n$
the level sets of Christoffel functions are known empirically to
provide tight
approximations to the support.
This support-approximating quality has motivated the use of Christoffel
functions in several statistical applications, such as density
estimation~\cite{lasserre2017empirical, lasserre2019empirical} and outlier
detection~\cite{askari2018kernel}. Additionally, the level sets have been
shown, using the plug-in approach~\cite{cuevas1997plug}, to converge exactly to
the support of the distribution (in the sense of
Hausdorff measure) when the degree of the polynomial approaches infinity and
when the true probability distribution is available~\cite{lasserre2019empirical}.
When the true probability distribution is \emph{not} known, as is typically the case in
data analysis, the Christoffel function can be empirically estimated using a
point cloud of independent and identically distributed (iid) samples from the
distribution: this \emph{empirical Christoffel function} still provides accurate
estimates for the support, and some convergence results in this case are also known~\cite{pauwels2021data}. 

In contrast to the asymptotic analysis of Christoffel functions reviewed above,
our interest is in developing error bounds that hold with a finite number of samples.
This paper is an extension of a conference 
paper~\cite{devonport2021datadriven} 
that reported our preliminary
work on support set estimation with polynomial Christoffel functions in the context of
data-driven reachability. 
In~\cite{devonport2021datadriven}, 
we investigated empirical inverse Christoffel functions constructed from iid
trajectory simulation data, and
provided a finite-sample guarantee of the probabilistic accuracy of reachable
set estimates produced by sublevel sets of this function.
The present paper significantly extends the
theory of finite-sample error bounds for support set estimators derived from
Christoffel functions by applying techniques from Bayesian PAC analysis,
a variation of classical PAC analysis that has been successfully applied to
Gaussian process classifiers~\cite{seeger2002pac}, kernel support vector
machines~\cite{langford2003pac}, and minimum-volume covering
ellipsoids~\cite{dolia2006minimum}.
This extension leverages a formal connection between the kernel empirical
inverse Christoffel function investigated by Askari
\emph{et al.}~\cite{askari2018kernel} and the
posterior variance of a Gaussian process regression model. In
conjunction with the PAC-Bayes theorem, the connection can be used to derive finite-sample
bounds for kernelized empirical inverse Christoffel functions. 

The application of Bayesian PAC analysis to the theory of Christoffel function
support set estimators has two benefits. First, it allows for the construction
of finite-sample guarantees for kernelized inverse Christoffel functions, which
to our knowledge have not been proved before. Second, when applied to polynomial
empirical inverse Christoffel function estimators, Bayesian PAC analysis can
provide guarantees of probabilistic accuracy and confidence with much greater
sample efficiency than the finite-sample bounds provided by classical VC
dimension bound arguments.

\subsection{List of Acronyms and Symbols}%
\label{sub:list_of_notation}

\quad

\begin{tabular}{p{0.17\textwidth}p{0.7\textwidth}}
\toprule
symbol & definition \\
\midrule
\multicolumn{2}{l}{
\emph{Reachability Analysis}
} \\
$\Phi(t_1;t_0,x_0,d)$ & State transition function, evolving a state
$x_0$ at time $t_0$ under disturbance $d$ to a state at time $t_1$\\
$\initialset$ & Set of initial states\\
$\distset$ & Set of disturbances\\
$t_0,t_1$ & Initial and final times\\
$\rs$ &  Forward reachable set\\
$\ars$ & Approximation of forward reachable set\\
\midrule
\multicolumn{2}{l}{
\emph{Probability, Statistical Learning Theory}
} \\
$\Ex{\cdot}$ & Expected value of a random variable\\
$\Pr{\cdot}$ & Probability of an event defined in terms of random variables\\
$D_{KL}(P||Q)$ & Kullback-Leibler (KL) divergence from $P$ to $Q$ \\
$D_{ber}(p||q)$ & KL divergence between Bernoulli
distributions with parameters $p$ and $q$ \\
$X$ & Random variable whose support we wish to estimate\\
$F_1$ & CDF of the chi-square distribution with 1 degree of freedom\\
$\dom$ & Domain of $X$ \\
$P_X$ & Probability measure of the distribution of $X$\\
$P_X^N$ & Probability measure of $N$ iid samples from $X$\\
PAC & Probably Approximately Correct\\
iid & Independent and Identically Distributed\\
$\epsilon$,$\delta$ & accuracy and confidence parameters in PAC guarantees\\
$\cclass$ & Concept class \\
$\bar{c}_Q$ & ``central concept'' of the posterior measure $Q$ \\
$P$,$Q$ & Prior and posterior probability measures on $\cclass$\\
$W_P$, $W_Q$ & Parametric representations of $P$ and $Q$\\
$C_P$, $C_Q$ & Stochastic estimators: random variables on $\cclass$ distributed
according to $P$, $Q$\\
$\ell(c,x)$ & statistical loss function comparing a concept $c$ and a
datum $x$\\
$r(c)$ & risk: average of $\ell(c,x)$ for $x\sim X$\\
$\hat{r}(c)$ & empirical estimate of $r(c)$ from data $x_1,\dotsc,x_N$\\
$\srisk$ & stochastic risk: average of $\ell(c,x)$ for $x\sim X$, $c\sim Q$\\
$\esrisk$ & empirical estimate of $\srisk$ from data $x_1,\dotsc,x_N$\\
\midrule
\multicolumn{2}{l}{
\emph{Christoffel Functions}
} \\
$M_m$, $\hat{M}_m$ & Matrix of moments of degree $\le m$ and its empirical
estimate\\
$\Mmat$ & Empirical moment matrix with diagonals modified by $\sigma_0$\\ 
$z_m(x)$ & vector of monomials with degree $\le m$ evaluated at point $x$\\
$\hat{\kappa}^{-1}(x)$ &  Polynomial empirical inverse Christoffel function
evaluated at $x$\\
$\hat{\kappa}^{-1}(x)$ &  kernelized empirical inverse Christoffel function\\
$C(x)$ & Christoffel-based support set estimator, output of
Algorithms~\ref{alg:cfun_classical},\ref{alg:cfun_pacbayes},
and~\ref{alg:cfun_pacbayes_poly}\\
\midrule
\multicolumn{2}{l}{
\emph{Gaussian Processes}
} \\
$m$, $k$ & prior mean and covariance functions\\
$m_q$, $k_q$ & posterior mean and covariance functions\\
$K$ & kernel Gramian matrix, $K_{ij}=k(x_i,x_j)$\\
$k_D$ & vector of kernel evaluations on data, $(k_D(x))_i=k(x_i,x)$\\
$\Norm{\mu,\Sigma}$ & Multivariate normal with mean $\mu$ and covariance
$\Sigma$\\
$\mathcal{GP}(m,k)$ & Gaussian process with mean and covariance functions $m$, $k$\\
\bottomrule
\end{tabular}

\section{Preliminaries}%
\label{sec:preliminaries}

\subsection{Probabilistic Reachability and Estimation of Support}%
\label{sub:probabilistic_reachability_and_estimation_of_support}

Consider a
dynamical system with a state transition function
$\Phi(t_1;t_0, x_0, d)$
that maps an initial state $x(t_0)=x_0\in\R^n$ at time $t_0$ to a unique final state
at time $t_1$, under
a disturbance $d:[t_0,t_1]\to\R^{w}$.
For instance, when the system state dynamics $ \dot{x}(t) = f(t,x(t),d(t))$
are known and have unique solutions on the interval $[t_0,t_1]$, then
$\Phi(t_1;t_0, x_0, d)$ is the solution of the state dynamics equation at time
$t_1$ with initial condition $x(t_0)=x_0$.
For the problem of forward reachability analysis, we are also given an
\emph{initial set} $\initialset\subset\R^n$, a set $\distset$ of allowed
disturbances and a time range $[t_0,t_1]$. The \emph{forward reachable set} is
then defined as the set of all states to which the system can transition in the
time range $[t_0,t_1]$ with initial states in $\initialset$ and disturbances in
$\distset$, that is the set
\begin{equation}
     \rs = \{\Phi(t_1;t_0,x_0, d) : x_0\in\initialset, d\in\distset\}.
\end{equation}

To tackle the problem of estimating the forward reachable set by statistical means, we
add probabilistic structure to the reachability problem 
by taking random variables $X_0$ and $D$ supported on $\initialset$ and $\distset$ respectively. 
These random variables then induce a random variable
$X=\Phi(t_1;t_0, X_0, D)$, 
whose support is precisely $\rs$ and
whose probability measure we denote as $P_X$.
A measure-theoretic interpretation $P_X(A)$ for a set $A$ is a measure of overlap between
$\rs$ and $A$: $P_X(A)$ is nonzero
only if $A$ has nonempty intersection with $\rs$, and $P_X(A)=1$ only if
$\rs\subseteq A$.
A probabilistic interpretation of $P_X(A)$ is that if we take samples $x_0$ and $d$ of the random
variables $X_0$ and $D$, then the vector $\Phi(t_1;t_0,x_0,d)$ lies in
$A$ with probability $P_X(A)$.
These interpretations motivate $P_X(A)$ as a
measure of \emph{probabilistic accuracy}: if a set $A\subseteq\R^n$ has a
greater measure $P_X(A)$ than a set $B\subseteq\R^n$, then $A$ is a more
accurate approximation of the reachable set than $B$, in the sense that it
``misses'' less of the probability mass than $B$ does.
In the probabilistic
version of the forward reachability problem, our goal is to find
reachable set approximations $\ars$ such that $P_X(\ars)$ is close to 1.
In addition, we will seek $\ars$ with low volume, in order to preclude trivial
estimates such as $\ars=\R^n$ and to generally minimize the conservatism of the
approximation.

The probabilistic relaxation of the forward reachability problem is a
statistical problem of \emph{support set estimation}
based on a finite set of observations. The support of a random variable is the
range of values it can assume: for example, if $X$ admits a probability density
function $p_X$, then the support of $X$ is the closure of the set
$\{x:p_X(x)\ne0\}$. 
In addition to the control-theoretic application developed above, support set
estimation has several applications in statistics and data science, such as
outlier and novelty detection~\cite{pauwels2021data,pauwels2016sorting,askari2018kernel}. It is therefore useful to consider the problem
for general random variables: we will do so for the theoretical developments in
this paper, returning to the reachability application in the numerical examples
of Section~\ref{sec:examples}. 
Formally, we address the following problem.
\begin{problem}\label{prb:pac_support}
    Given accuracy and confidence parameters $\epsilon,\delta\in(0,1)$
    and
    a random variable $X$ whose support lies in a compact domain $\dom\subseteq \R^n$,
    collect data $x_1,\dotsc,x_N\simiid X$
    and use them to find a set $c(\epsilon,\delta;x_1,\dotsc,x_N)\subset\dom$ such that 
    the following bound holds:
    \begin{equation}
        \label{eq:prob_pac_bound}
        P_X^N\left(\{x_1,\dotsc,x_N:P_X(c(\epsilon,\delta;x_1,\dotsc,x_N)) \ge 1-\epsilon \}\right) \ge 1-\delta.
    \end{equation}
\end{problem}
The bound~\eqref{eq:prob_pac_bound} is known as a Probably Approximately Correct
(PAC) bound, which appears frequently in statistical learning theory.
The two probability inequalities in~\eqref{eq:prob_pac_bound} are interpreted as
assertions of probabilistic accuracy and confidence:
\begin{itemize}
    \item \emph{accuracy}: the inner inequality 
        $P_X(c(\epsilon,\delta;x_1,\dotsc,x_N)) \ge 1-\epsilon$
        asserts that the probabilistic accuracy of 
        the estimator is at least $1-\epsilon$.
    \item confidence: the outer inequality 
        asserts that the accuracy statement holds with probability $1-\delta$
        with respect to $P_X^N$. The probability, and hence the confidence, 
        is with respect to the data: $P_X^N$ is the probability measure
        corresponding to $N$ iid observations drawn from $X$, so 
        $P_X^N(A)$ for $A\subseteq\dom^N$ denotes the probability that 
        $x_1,\dotsc,x_N\in A$. 
        Thus the inequality
        $P_X^N(\{x_1,\dotsc,x_N: \cdots \}) \ge 1-\delta$
        asserts that the observed data set $x_1,\dotsc,x_N$ belongs, with probability
        at least $1-\delta$, to the class of
        data sets sufficiently informative to yield an estimator $c$ satisfying the
        accuracy assertion.
\end{itemize}
For brevity, we drop the arguments of the estimator
$c(\epsilon,\delta;x_1,\dotsc,x_N)$ 
from the notation, understanding that an estimator $c$ is always constructed
using a given set of data $x_1,\dotsc,x_N$, with respect to given parameters
$\epsilon$ and $\delta$.
The sample size $N$ is a fixed problem parameter: indeed, finding a suitable $N$ is part of
solving the problem.
In addition to the requirements given in Problem~\ref{prb:pac_support}, 
we may also impose that the estimator $c$ be
drawn from a pre-specified class of admissible estimators. Such a condition
allows us to restrict attention to computationally feasible sets, or sets with
certain properties such as compactness for cases when the reachable set is known
to be compact.
In classical PAC analysis, the structure of the pre-specified class also
plays a key role in determining an appropriate $N$.

\subsection{Christoffel Functions}%
\label{sub:christoffel_functions}

Given a finite measure $P_X$ on $\R^n$ and a positive integer $\PCord$,
the Christoffel function of order $m$ is defined as the ratio
$
    \kappa(x) = 1/z_m(x)^\top M_{m}^{-1}z_m(x),
$
where $z_m(x)$ is the vector of monomials of degree $\le m$, and
where $M_{m}$ is the matrix of moments
$
    M_m = \int_\dom z_m(x) z_m(x)^\top dP_X(x).
$
We assume throughout that $M_m$ is positive definite, ensuring that $M_m^{-1}$
exists. 
The Christoffel function has several important applications in
approximation theory~\cite{nevai1986geza}, where its asymptotic properties are used to prove the
regularity and consistency of Fourier series of orthogonal
polynomials~\cite{xu1995christoffel}. 
For our
purposes, it is more convenient to use the \emph{inverse Christoffel function}
$
    {\kappa(x)}^{-1} = z_m(x)^\top M_{m}^{-1}z_m(x),
$
which is a polynomial of degree $2m$.
In Problem~\ref{prb:pac_support}, and more generally in the problem of estimating a
probability distribution from samples, $P_X$ is unknown. 
In this case, 
we instead use an empirical estimate for
the moment matrix $M_{m}$, namely
$
\hat{M}_m
=
\frac{1}{N}
\sum_{i=1}^N
z_m(x_i)
z_m(x_i)^\top
$.
The matrix $\hat{M}_m$ is positive semidefinite: it is additionally
positive definite, and hence nonsingular, if $N \ge \binom{n+m}{n}$ and
$x_1,\dotsc,x_N$ do not all belong to the zero set of a single degree $m$
polynomial. 
It is useful, both numerically and theoretically, to modify this empirical
estimate adding a scaled identity perturbation: thus we take
$
\hat{M}_{m,\sigma}
= 
\sigma^2 I 
+ 
\frac{1}{N}
\sum_{i=1}^N
z_m(x_i)
z_m(x_i)^\top,
$
as our empirical moment matrix in the sequel,
where $\sigma^2 > 0$ is a term fixing the magnitude of the
perturbation. 
In addition to its role in developing the kernel extension, 
the $\sigma^2I$ term generally improves the conditioning of the empirical moment
matrix and ensures nonsingularity in all cases.
The empirical moment matrix $\hat{M}_{m,\sigma}$ itself defines a Christoffel
function, whose inverse
\begin{equation}
    \label{eq:empirical_cfun_def}
    \hat{\kappa}^{-1}(x)=z_m(x)^\top 
        \hat{M}_{m,\sigma}^{-1}z_m(x)
\end{equation}
is called the \emph{empirical inverse Christoffel function}.

The dyadic sum $\frac{1}{N}\sum_{i=1}^N z_m(x_i)z_m(x_i)^\top$ can be expressed
as the matrix product $\frac{1}{N} Z Z^\top$, where
$Z\in\R^{\binom{n+m}{n}\times N}$ is the matrix
    $
    Z =
    \begin{bmatrix}
        z_m(x_i) & \hdots & z_m(x_N)
    \end{bmatrix}
    $
of polynomial features.
By expressing the dyadic sum this way, we can apply
the matrix inversion lemma to express the inverse of the empirical moment
matrix
as
\begin{equation}
    \label{eq:empirical_moment_matrix_mil}
    \hat{M}_{m,\sigma} = 
    \left(
        \sigma^2 I 
        + 
        \tfrac{1}{N}
        Z Z^\top
    \right)^{-1}
    =
    \sigma^{-2}\left(I - Z\left(\sigma^2 N I + Z^\top
    Z\right)^{-1}Z^\top\right).
\end{equation}
This expression for $\hat{M}_{m\sigma}$ allows us to rewrite
the empirical inverse Christoffel function as
\begin{equation}
    \label{eq:pre_kernel}
    \hat{\kappa}^{-1}(x)
    =
    N
    \sigma_0^{-2} z_m(x)^\top z_m(x)
    -
    N
    \sigma_0^{-2} z_m(x)^\top Z \left(\sigma_0^2 I + Z^\top Z\right)^{-1}Z^\top z_m(x),
\end{equation}
where we have made the change of variables $\sigma^2=\sigma_0^2/N$.
The vector $z_m$ enters~\eqref{eq:pre_kernel} only through 
the inner products $z_m(x_i)^\top z_m(x_j)$:
The matrix $Z^\top Z\in\R^{N\times N}$ has elements 
$(Z^\top Z)_{ij}=z_m(x_i)^\top z_m(x_j)$, and the matrix-vector product
$Z^\top z_m(x)$ has elements 
$(Z^\top z_m(x))_i=z_m(x_i)^\top z_m(x)$. 
By replacing the inner product
$z_m(x_i)^\top z_m(x_j)$ with
an arbitrary positive definite%
\footnote{Here, and throughout the paper, we mean positive definite in the
sense of reproducing kernel Hilbert spaces and kernel machines, which is
that a square matrix $K$ with elements $(K)_{ij}=k(x_i,x_j)$ is a positive definite
matrix.}
function $k:\R^{n}\times\R^{n}\to\R$ and rescaling by a factor of
$\sigma_0^{2}/N$, we obtain the kernelized variant of the empirical inverse
Christoffel function,
\begin{equation}
    \label{eq:kernel_cfun_def}
    \kappa^{-1}(x)
    =
    k(x,x) - k_D(x)^\top\left(\sigma_0^2 I + K\right)^{-1}k_D(x),
\end{equation}
where $K\in\R^{N\times N}$ and $k_D(x)\in\R^N$ are defined as
\begin{equation}
    \begin{aligned}
        \label{eq:gramian_and_kvec}
        K_{ij} &=k(x_i,x_j), &
        (k_D(x))_i &=k(x_i,x).
    \end{aligned}
\end{equation}

\section{Christoffel Function Estimators of Support}%
\label{sec:christoffel_function_estimators_of_support}

\begin{wrapfigure}{R}{0.5\textwidth}
    \begin{minipage}[t]{0.46\textwidth}
    \begin{algorithm}[H]
        \caption{To estimate a support set by a polynomial empirical inverse
        Christoffel function satisfying a classical PAC bound.}\label{alg:cfun_classical}
        \begin{algorithmic}
        \STATE inputs: random variable $X$ with support in $\dom$; polynomial order $m\in\N_+$; PAC parameters
        $\epsilon,\delta\in(0,1)$; noise parameter $\sigma_0^2\in\R_{++}$;
        \STATE $N\gets \ceil{\frac{5}{\epsilon}\left(\log\frac{4}{\delta}+\binom{n+2m}{n}\log\frac{40}{\epsilon}\right)}$
        \FOR {$i\in\{1,\dotsc,N\}$} 
          \STATE sample $x_i\sim X$
        \ENDFOR
        \STATE
        $\Mmat\gets \sigma_0^2I + \frac{1}{N}\sum\limits_{i=1}^N z_m(x_i)z_m(x_i)^\top$
        \STATE
        $\alpha \gets \max_i z_m(x_i)^\top\Mmat^{-1} z_m(x_i)$
        \STATE
        $C(x)=z_m(x)^\top\Mmat^{-1}z_m(x)$;
        \RETURN $\ind{C(x)\le\alpha}$;
        \end{algorithmic}
    \end{algorithm}
\end{minipage}
\end{wrapfigure}

Algorithms~\ref{alg:cfun_classical} and~\ref{alg:cfun_pacbayes} are procedures
to estimate the support of a random variable with a sublevel set of an
empirical inverse Christoffel function, where the only information needed from
the random variable is a collection of iid samples. 
Algorithm~\ref{alg:cfun_classical} is designed to satisfy a classical PAC bound.
This has the advantages of providing an \emph{a priori} sample bound, and of
admitting a fairly direct proof, which is given in 
Section~\ref{sub:classical_pac_analysis_of_christoffel_function_estimators}. The
essence of the proof is to demonstrate that the sublevel sets of a polynomial empirical
inverse Christoffel function of a given order inhabit a concept class of known
VC dimension. This argument is valid for polynomial Christoffel functions of any
order,
but it is generally not valid for kernelized Christoffel functions. Indeed, the
classes of sublevel sets of certain kernelized empirical inverse Christoffel
functions can have infinite VC dimension, so a classical PAC bound is not
possible in general for kernelized empirical inverse Christoffel functions.
Algorithm~\ref{alg:cfun_pacbayes} is designed to satisfy a Bayesian PAC bound
which is developed in
Section~\ref{sub:bayesian_pac_analysis_of_christoffel_function_estimators}.
Unlike the classical PAC bound provided for Algorithm~\ref{alg:cfun_classical},
this Bayesian PAC bound is applicable to all kernelized empirical inverse
Christoffel functions, including those whose sublevel sets have infinite VC
dimension. 
When applied to 
polynomial empirical inverse Christoffel
functions as a special case, we find that it is more sample-efficient than the
classical PAC bound:
in some of the examples in Section~\ref{sec:examples}, the
Bayesian PAC bound requires an order of magnitude fewer samples to achieve the
same accuracy and confidence as that guaranteed by the classical PAC bound.
The disadvantages of the Bayesian PAC approach is that the required number of
samples is not known \emph{a priori}, since certain terms in the bound depend on
the data. Algorithm~\ref{alg:cfun_pacbayes} therefore takes an iterative
approach, taking samples in batches and re-evaluating the Bayesian PAC bound
after each batch until it reaches the desired level of accuracy.

\begin{remark}
    \label{rmk:rs_subset}
    In some reachability problems, we are only interested in computing a
    reachable set for a subset of the state variables.
    For example, suppose the state is $(x_1,\dotsc,x_n)\in\R^n$, and we wish to
    verify a safety specification involving only the states $x_1,\dotsc,x_s$,
    where $s<n$: a reachable set for the states $x_1,\dotsc,x_s$ would suffice
    for this problem.
    In cases like this, the algorithms presented in this section can be modified
    to use only the first
    $s$ elements of the samples. The output of the algorithm is then
    an empirical inverse Christoffel function with domain $\R^s$ whose sublevel
    set $\ars$ estimates the reachable set for the reduced set of states.
    In the sequel, we refer to this variation of the algorithms in this
    section as their \emph{reduced-state} variations.
\end{remark}

\subsection{Classical PAC Analysis}%
\label{sub:classical_pac_analysis_of_christoffel_function_estimators}

PAC bounds originate in study of empirical
risk minimization problems in statistical learning theory. Our strategy to prove a PAC bound for
Algorithm~\ref{thm:cfun_pac} is to express Problem~\ref{prb:pac_support} as an
empirical risk minimization problem and to then apply the tools of statistical
learning theory.

In empirical risk minimization, the objective is to match a concept
$c\subseteq\dom$ from a pre-specified concept class
$\cclass\subseteq 2^\dom$ 
to an
unknown random variable $X$ supported on $\dom$
using only a finite set of iid observations $x_1,\dotsc,x_N$ of $X$.
How well a concept
matches $X$ is quantified by the statistical risk $r(c)=\Ex{\ell(c,X)}$ defined
by a loss function $\ell:\cclass\times\dom\to\R_+$ and the unknown measure
$P_X$: a lower risk indicates a better match.
Since we do not know $P_X$, we cannot
directly evaluate the statistical risk. However, we can use
the empirical risk
$\erisk=\avgion \ell(c,x_i)$ as a proxy for the true risk, and select a concept
to match the data on the basis of minimizing the empirical risk. 

Whether empirical risk minimization actually selects a concept with low risk
depends on how much $\erisk$ differs from $\risk$. A classical PAC bound
provides a bound on the difference $\risk-\erisk$, or the absolute difference,
that holds with high probability.
We use the following result from~\cite{alamo2009randomized}, which
gives a quantitative sample bound that depends on the
Vapnik-Chervonenkis (VC) dimension~\cite{vidyasagar2002learning} of the concept
class. The VC dimension of a concept is a combinatorial measure of its
complexity based on the expressiveness of its concepts. 
\begin{lemma}[\cite{alamo2009randomized}, Corollary 4]
    \label{lem:generic_pac}
    Let $\cclass$ be a concept class of sets with VC dimension $\le d$, and let
    $\ell:\cclass\times\dom\to\{0,1\}$ denote a $\{0,1\}$-valued loss function.
    If
    \begin{equation}
        \label{eq:generic_pac_bound}
        N \ge \frac{5}{\epsilon}\left( \log\frac{4}{\delta} + d \log\frac{40}{\epsilon} \right),
    \end{equation}
    and if $\hat{r}(c)=0$, then $P_X^N\left(\{x_1,\dotsc,x_N : r(c) \le \epsilon\}\right) \ge 1-\delta$.
\end{lemma}
A concept class with higher VC dimension generally provides greater-fidelity
estimates than one with lower VC dimension, but is also more prone to
overfitting: informally, this is the reason why a concept class with higher VC
dimension requires a larger sample bound for the same accuracy and confidence
than one with lower VC dimension.

To apply Lemma~\ref{lem:generic_pac}, we must show that the sublevel sets of a
polynomial empirical inverse Christoffel function belong to a concept class
of bounded VC dimension. 
One such class is the class of superlevel sets of degree $2k$ polynomials: the
following Lemma from~\cite{dudley1978central}, provides a bound on the
VC dimension. 

\begin{lemma}[\cite{dudley1978central}, Theorem 7.2]
    \label{lem:pos_vc}
    Let $V$ be a vector space of functions $g:\R^n\to\R$ with dimension $d$.
    Then the class of sets
        $
        \text{Pos}(V) = \left\{\ \{x : g(x)\ge 0\}, g\in V\right\}
        $
    has VC dimension $\le d$.
\end{lemma}

The PAC bound, and hence the validity of Algorithm~\ref{alg:cfun_classical}
follows from Lemmas~\ref{lem:pos_vc} and~\ref{lem:generic_pac} by framing the
support estimation problem as one of empirical risk minimization.
\begin{theorem}
    \label{thm:cfun_pac}
    The support set estimate produced by Algorithm~\ref{alg:cfun_classical},
    that is the set $\{x\in\dom: C(x) \le \alpha\}$
    where $C(x)=z_m(x)^\top\Mmat^{-1}z_m(x)$,
    $\alpha=\max_i C(x_i)$,
    satisfies the PAC bound
        $
        P_X^N(\{x_1,\dotsc,x_N:P_X(\{x\in\dom: C(x) \le \alpha\}) \ge
        1-\epsilon\})\ge 1-\delta,
        $
    and thereby solves Problem~\ref{prb:pac_support} with parameters
    $\epsilon,\delta$.
\end{theorem}

\begin{proof}
Let
$\mathcal{C}=\text{Pos}(\R[x]_{2m}^n)$,
and 
$\ell(c,x)=\ind{x\notin c}$.
Note that the set 
$\{x\in\R^n : C(x) \le \alpha\}$ 
is a member of
$\text{Pos}(\R[x]_{2m}^n)$, 
since it can be expressed as 
$c=\{x\in\R^n : \alpha - C(x) \ge 0\}$.
Since the dimension of $\R[x]_{2m}^n$ is $\binom{n+2m}{n}$, the
VC dimension of 
$\text{Pos}(\R[x]_{2m}^n)$ 
is $\le\binom{n+2m}{n}=d$ by Lemma~\ref{lem:pos_vc}. 
For 
$\ell(c,x)=\ind{x\notin c}$,
the statistical risk is 
$r(c)=\Ex{\ind{x\notin c}}=1-P_X(c)$,
and its empirical counterpart is
$\hat{r}(c)=\sum_{i=1}^N \ind{x_i\notin c}$.
The empirical risk is zero for any set $c$ that encloses 
$x_1,\dotsc,x_N$. The set 
$\{x\in\R^n : C(x) \le \alpha\}$ 
encloses 
$x_1,\dotsc,x_N$ 
by construction, meaning that 
$\hat{r}(\{x\in\R^n : C(x) \le \alpha\})=0$.
By applying Lemma~\ref{lem:generic_pac} for this choice of $\mathcal{C}$,
$\ell$, and $m$, we find that if
$
N \ge
    \frac{5}{\epsilon}\left(
    \log\frac{4}{\delta} + \binom{n+2m}{n} \log\frac{40}{\epsilon}
\right),
$
then
$P_X^N\left(\{x_1,\dotsc,x_N\} : 1-P_X(\{x\in\R^n : C(x) \le \alpha\}) \le \epsilon\}\right) \ge 1-\delta$.
Since Algorithm~\ref{alg:cfun_classical} selects $N$ to be the smallest integer
such that
$
N \ge
    \frac{5}{\epsilon}\left(
    \log\frac{4}{\delta} + \binom{n+2m}{n} \log\frac{40}{\epsilon}
\right),
$
it follows that the stated PAC bound holds for the output of
Algorithm~\ref{alg:cfun_classical}.
\end{proof}

\subsection{Bayesian PAC Analysis}%
\label{sub:bayesian_pac_analysis_of_christoffel_function_estimators} 

\quad

\begin{wrapfigure}{R}{0.5\textwidth}
\begin{minipage}[t][11cm][t]{0.46\textwidth}
    \begin{algorithm}[H]
        \caption{To estimate a support set by a kernelized empirical inverse
        Christoffel function satisfying a Bayesian PAC bound.}\label{alg:cfun_pacbayes}
        \begin{algorithmic} 
        \STATE inputs: random variable $X$ with support in $\dom$; positive definite kernel function $k$; PAC parameters
        $\epsilon,\delta\in(0,1)$; noise parameter $\sigma_0^2\in\R_{++}$;
        initial sample size $N_0$; batch size $N_b$; threshold $\eta$.
        \STATE $N\gets N_0$
        \STATE $D\gets (x_1,\dotsc,x_{N})\simiid X$
        \STATE $i\gets0$
        \STATE $\epsilon^0\gets 1$
        \WHILE {$\epsilon^i > \epsilon$}
            \STATE $i\gets i + 1$
            \STATE append \\ $(x_{N+1},\dotsc,x_{N+N_b})\simiid X$ to $D$
            \STATE $N\gets N+N_b$
            \STATE $K_{\sigma_0}\gets\sigma_0^2I + K$
            \STATE define $C:\dom\to\R_{+}$ to be
            $C(x)=k(x,x)-k_D(x)K_{\sigma_0}^{-1}k_D(x)$;
            \STATE Evaluate $\overline{r}$ as in~\eqref{eq:pacbayes_kernel}
            \STATE 
            $\epsilon_i \gets \frac{\bar r + \frac{2}{N}\log (\frac{\pi^2i^2}{6
            \delta})}{1-F_1(1)}$, $F_1$ as in \eqref{eq:f1_def}
        \ENDWHILE 
        \RETURN $\ind{C(x)\le\eta}$
        \end{algorithmic}
    \end{algorithm}
\end{minipage}
\end{wrapfigure}

Bayesian PAC analysis bounds the deviation of
the expected values of the true and empirical risks with respect to a
data-dependent probability measure. 
Given a
prior measure $P$ over $\cclass$ and a posterior measure $Q$ derived from the
prior and the observations, we define the expected risk $\srisk = \Ex{\ell(c, X)}$ and
empirical expected risk $\esrisk=\Ex{\avgion \ell(c, x_i)}$
where $c\sim Q$.
Equivalently, $P$ and $Q$ define random variables $C_P$, $C_Q$ supported on
$\cclass$, called the
prior and posterior \emph{stochastic estimators}: $\srisk$ and $\esrisk$ are the
true and empirical risks of $C_Q$.
A Bayesian PAC bound is a bound on the deviation between $\srisk$ and $\esrisk$.
Bayesian PAC bounds can be used to provide an error bound for a
single classifier which captures the central behavior of $Q$, which we call the
\emph{central concept} and denote as $\bar{c}_Q$. 
To verify that Algorithms~\ref{alg:cfun_pacbayes}
provides a valid solution to
Problem~\ref{prb:pac_support}, we show that 
its output is the central concept of a posterior stochastic estimator and use a
Bayesian PAC bound to show that a bound of the 
form~\eqref{eq:prob_pac_bound} holds.

The most common tool to construct Bayesian PAC bounds is the
PAC-Bayes theorem developed by McAllester~\cite{mcallester1999some}, 
Seeger~\cite{seeger2002pac} and others~\cite{langford2005tutorial}. 
We use the variation due to Seeger. 
This theorem assumes that the concept class 
admits a parameterization which can be infinite-dimensional.
\begin{theorem}[PAC-Bayes Theorem, adapted from~\cite{seeger2002pac,langford2005tutorial}]
    \label{thm:pacbayes}
    Consider a concept class $\cclass$ admitting a parametrization by
    $w\in\mathcal{W}$.
    Let the loss function be zero-one valued, that is
    $\ell:\cclass\times\dom\to\{0,1\}$.
    The following bound holds for all measures $P$, $Q$ over the concept class $\cclass$
    defined by measures $W_P$ and $W_Q$ over $\mathcal{W}$
    such that $W_Q$ is absolutely continuous with respect to $W_P$:
    \begin{equation}
        \label{eq:pacbayes}
        P_X^N\left( \left\{x_1,\dotsc,x_N:
            D_{ber}(\esrisk || \srisk) 
            \le 
            \frac{
                D_{KL}(W_Q || W_P) + \log\frac{N+1}{\delta}
            }{N}
        \right\}\right) \ge 1 - \delta.
    \end{equation}
\end{theorem}
Here, $D_{KL}(W_P||W_Q)$ denotes the Kullback-Leibler (KL) divergence between
$W_P$ and $W_Q$, and $D_{ber}(q||p)$ denotes the KL divergence between two
Bernoulli distributions with parameters $q$ and $p$, given by the formula
$D_{ber}(q||p) = q\log\frac{q}{p} + (1-q)\log\frac{1-q}{1-p}$.
For a given set of data $x_1,\dotsc,x_N$, confidence parameter $\delta$, 
and a prior measure $P$ chosen independently of the data, 
the inequality~\eqref{eq:pacbayes} provides a family of Bayesian PAC bounds, one
for each posterior measure $Q$.

We use the PAC-Bayes theorem in the proof of
Theorem~\ref{thm:alg_cfun_pacbayes}, 
which asserts the validity of
Algorithm~\ref{alg:cfun_pacbayes}.
First, we construct prior and
posterior stochastic estimators $C_P$ and $C_Q$, corresponding to measures $P$,
$Q$ over a concept class,
which admit a sublevel set of the empirical inverse Christoffel function as a
central concept; namely
$\bar{c}_Q=\{x:\kappa^{-1}(x) \le\eta\}$ for a given positive $\eta$. 
Next, we express a formula to compute the empirical stochastic risk $\esrisk$
of $C_Q$ from the data. 
Then, we establish a bound on the true stochastic $\srisk$ in terms of $\esrisk$ using the PAC-Bayes theorem.
Finally, we prove a bound on the true risk $r(\bar{c}_Q)$ of the central concept
in terms of $\srisk$. This sequence of bounds combines to yield
a bound of the
form~\eqref{eq:prob_pac_bound} computable in terms of known data. 

\begin{theorem}
    \label{thm:alg_cfun_pacbayes}
    Denote $C^i$ as the inverse Christoffel function constructed during the $i$th
    iteration of Algorithm 3.2. We have the following PAC bound on all the inverse
    Christoffel functions constructed during the algorithm:
    \begin{equation*}
        \begin{aligned}
        \label{eq:pac_epsi}
        &\mathbb P(\forall i\geq1,\, P_X(\{x:C^i(x) \leq\eta\}) \geq 1-\epsilon^i)\\
        &\quad \geq 1-\delta.
        \end{aligned}
    \end{equation*}
    Thus, with confidence $\delta$, upon the termination condition of Algorithm
    3.2, we are left with a support set estimate of probability mass $\geq 1- \epsilon$.
\end{theorem}

In addition to verifying the validity of the terminal output of
Algorithm~\ref{alg:cfun_pacbayes}, Theorem~\ref{thm:alg_cfun_pacbayes} justifies
the use of Algorithm~\ref{alg:cfun_pacbayes} in an ``any time algorithm''
fashion, that is as an algorithm whose output is verified even if execution is
stopped prematurely. The execution of Algorithm~\ref{alg:cfun_pacbayes} will terminate as long as the growth of 
$D_{KL}(\Norm{0, (\sigma_0^{-1}I + K^{-1})^{-1}}||\Norm{0,K})$ is $o(N)$:
determining the conditions under which this growth condition holds is a topic
for future research.
 
We now develop the constructions used in the proof, starting with the
prior and posterior stochastic estimators for the kernel case.
We take
\begin{equation}
    \label{eq:kernel_prior_posterior_families}
    \begin{aligned}
        C_{P}
        &= \{x : g_p(x)^2 \le \eta\},
        & C_{Q} 
        &= \{x : g_q(x)^2 \le \eta\},
    \end{aligned}
\end{equation}
where $g_p$ and $g_q$ are the prior and posterior of a general Gaussian process
regression model with prior kernel $k$, conditioned on the observations
$x_1,\dotsc,x_N$, $y_1=\dotso=y_N=0$ with observation noise level $\sigma_0^2$.
The corresponding concept class is the class of $\eta$-sublevel sets of
functions in the support of $g_p$, which depends on the choice of kernel.
According to~\eqref{eq:gp_kernel}, $g_q$ has
posterior mean $m_q=0$ and variance
\begin{equation}
        \Var[g_q]{x}= k(x,x) - k(X,x)^\top{\left(\sigma^2 I_N + K(X,X)\right)}^{-1} k(X,x).
\end{equation}
We take the posterior central concept to be
$\bar{c}_Q=\{x : \Ex{g_q(x)^2} \le\eta\}$. 
Since $\Ex{g_q(x)}=m_q(x)=0$ for all $x\in\dom$, we know
$\Ex{g_q(x)^2}=\Var[g_q]{x}$. This means that the posterior central concept is
\begin{equation}
    \begin{aligned}
        \label{eq:kernel_central_concept}
        \bar{c}_{Q} &= \{x :
            k(x,x) - k(X,x)^\top{\left(\sigma^2 I_N + K(X,X)\right)}^{-1} k(X,x) 
            \le \eta
        \}
        = \{x : \kappa^{-1}(x) \le \eta\}
    \end{aligned}
\end{equation}
as desired.

Next, we construct the sequence of bounds, starting with the formula for the empirical
stochastic risk of $C_Q$ in terms of known data. 
\begin{lemma}
    \label{lem:empirical_stochastic_risk}
    For the zero-one membership loss $\ell(c,x)=\ind{x\notin c}$, the empirical
    stochastic risk of the posterior stochastic
    estimators $C_Q$ defined in~\eqref{eq:kernel_prior_posterior_families} is
    \begin{equation}
    \label{eq:poly_empirical_stochastic_risk}
        \esrisk
        =
        \avgion 1-F_1\left(\frac{\eta}{\kappa^{-1}(x_i)}\right),
    \end{equation}
    where $F_1$ is the CDF of the chi-square distribution with one degree of
    freedom, that is 
    \begin{equation}
        \label{eq:f1_def}
        F_1(x)=\Pr{Z^2 \le x} \text{ where } Z\sim\Norm{0,1}.
    \end{equation}
\end{lemma}

Next, we use the PAC-Bayes theorem to bound the stochastic risk $\srisk$ by the
empirical stochastic risk $\esrisk$. 

\begin{lemma}
    \label{lem:sriskbound_kernel}
   Let $x_1,\dotsc,x_N\simiid X$ denote a set of observations used to construct
   $C_{Q}$ from $C_{P}$
   in~\eqref{eq:kernel_prior_posterior_families}. The stochastic risk $\srisk$
   is bounded by $\overline{r}\in(0,1)$, 
   where
   \begin{equation}
       \label{eq:pacbayes_kernel}
       \overline{r}=\sup
       \left\{ \beta :
       D_{\text{ber}}(\esrisk || \beta)
       \le
       \frac{D_{KL}(\Norm{0,(K^{-1}+\sigma_0^{-2}I)^{-1}} || \Norm{0,K})
       + \log\frac{N+1}{\delta}
       }{N}\right\},
   \end{equation}
   with confidence $1-\delta$.
\end{lemma}
Since $D_{ber}(q||p)$ is convex in $(q,p)$ and equal to zero for $q=p$, the set
in~\eqref{eq:pacbayes_kernel} is an interval containing $\esrisk$.
Once $\esrisk$ and the right-hand side of the inequality
in~\eqref{eq:pacbayes_kernel} are evaluated, the supremum $\overline{r}$ can be
computed using a scalar root-finding procedure to solve
$
D_{\text{ber}}(\esrisk || \beta)
-
(D_{KL}(\Norm{0,(K^{-1}+\sigma_0^{-2}I)^{-1}} || \Norm{0,K})
+ \log\frac{N+1}{\delta}
)/N = 0
$ over the interval $\beta\in[\esrisk,1)$.

Finally, we relate the statistical risk of $r(\bar{c}_Q)$ to $\srisk$.

\begin{lemma}
    \label{lem:central_concept_risk_bound}
    The statistical risk $\risk[\bar{c}_\eta]$ of the posterior central concept
    and the stochastic risk $\srisk$ of the posterior stochastic estimator
    satisfy the
    bound
        $
        r(\bar{c}_Q) \le \frac{1}{1-F_1(1)}\srisk\approx 3.15\srisk.
        $
\end{lemma}

When combined, the sequence of bounds, the sequence of bounds above provide a
bound of the form~\eqref{eq:prob_pac_bound} that holds independently for each
iteration of Algorithm~\ref{thm:alg_cfun_pacbayes}. Applying a union bound
argument to provide a guarantee that holds uniformly over iterations
forms the central argument of the proof of Theorem~\ref{thm:alg_cfun_pacbayes}.

\begin{proof}[Proof (of Theorem~\ref{thm:alg_cfun_pacbayes})]
    The bound is trivially satisfied at the beginning of execution, since
    $\epsilon^0\gets1$. Next, let $i>0$, and let
    $C^i_Q$ denote the stochastic classifier $\{g^i_Q(x)^2 \leq
    \eta\}$, where $g^i_Q(x) \sim \mathcal N(0, k(x,x) - k_{D^i}(x)^\top
    (\sigma_0^2 I + K^i) k_{D^i}(x))$, with the $i$ superscripts signifying
    using the dataset accumulated so far at iteration $i$. Let $r^i_Q$ denote
    the risk of $C^i_Q$.
    By Lemma 3.7, we have
        $\forall i\geq 1,\; \mathbb P(r^i_Q > (1-F_1(1))\epsilon^i) \leq \frac{6\delta}{\pi^2 i^2}$.
        By a union bound,
    $ \mathbb P(\exists i,\, r_Q^i > (1-F_1(1))\epsilon^i) \leq \sum_{i\geq 1} \frac{6\delta}{\pi^2 i^2} = \delta$.
    Thus, with probability at least $1-\delta$, every $r_Q^i \leq \epsilon^i$.
    On this event, by Lemma 3.8, we have 
    $
        \forall i \geq 1,\, P_X(\{x: C^i(x) > \eta\}) \leq \frac{r_Q^i}{1-F_1(1)} = \epsilon^i
    $
    as desired. 
\end{proof}

\subsection{Bayesian PAC Analysis: the Polynomial Case}%
\label{sub:bayesian_pac_analysis_the_polynomial_case}

\begin{wrapfigure}{r}{0.5\textwidth}
\begin{minipage}[t]{0.46\textwidth}
    \begin{algorithm}[H]
        \caption{To estimate a support set by a polynomial empirical inverse
        Christoffel function satisfying a Bayesian PAC bound.}\label{alg:cfun_pacbayes_poly}
        \begin{algorithmic}
        \STATE inputs: random variable $X$ with support in $\dom$; Christoffel function order $m$; PAC parameters
        $\epsilon,\delta\in(0,1)$; noise parameter $\sigma_0^2\in\R_{++}$;
        initial sample size $N_0$; batch size $N_b$;
        \STATE $N\gets N_0$
        \STATE $D\gets (x_1,\dotsc,x_{N})\simiid X$
        \STATE $i\gets0$
        \STATE $\epsilon^0\gets 1$
        \WHILE {$\epsilon^i > \epsilon$}
            \STATE $i\gets i + 1$
            \STATE append \\ $(x_{N+1},\dotsc,x_{N+N_b})\simiid X$ to $D$
            \STATE $N\gets N+N_b$
            \STATE define $C:\dom\to\R_{+}$ to be \\
            $C(x)=z_m(x)^\top\Mmat^{-1}z_m(x)$;
            \STATE evaluate $\overline{r}$ as in~\eqref{eq:pacbayes_poly}
            \STATE
            $\epsilon_i \gets \frac{\bar r + \frac{2}{N}\log (\frac{\pi^2i^2}{6
            \delta})}{1-F_1(1)}$, $F_1$ as in \eqref{eq:f1_def}
        \ENDWHILE
        \RETURN $\ind{C(x)\le\eta}$
        \end{algorithmic}
    \end{algorithm}
\end{minipage}
\end{wrapfigure}

With the general kernel case settled, we now consider the polynomial case in
particular. Since the kernel case reduces to the polynomial case by the kernel
$k(x,y)=z_m(x)^\top z_m(y)$, we have in a sense already provided a bound for the
polynomial empirical inverse Christoffel function by means of Bayesian PAC
analysis. However, we can construct a prior and posterior stochastic estimator
for the polynomial case which avoids direct use of the $N\times N$ kernel
Gramian, which can be computationally advantageous. The special prior and
posterior stochastic estimators are
\begin{equation}
    \begin{aligned}
        \label{eq:poly_cfun_prior_family}
        C_{P} &= \{x: (W_P^\top z_m(x))^2 \le \eta\},
        \\
        C_{Q} &= \{x: (W_Q^\top z_m(x))^2 \le \eta\},
    \end{aligned}
\end{equation}
where
$W_P\sim\Norm{0,\sigma_0^{-2} I}$,
$W_Q\sim\Norm{0, \Mmat^{-1}}$.

Notice that $W_P^\top z_m$ and $W_Q^\top z_m$ are Gaussian processes: indeed,
they correspond to
the prior and posterior of a general Gaussian process
regression model with prior kernel $k(x,y)=z_m(x)^\top z_m(y)$, conditioned on
the observations $x_1,\dotsc,x_N$, $y_1=\dotso=y_N=0$ with observation noise
level $\sigma_0^2$.
We take the central concept $\bar{c}_{Q}$ of $C_{Q}$ to be the $\eta$-sublevel set 
\begin{align}
    \bar{c}_{Q} 
    &= \{x: \Ex{(W_Q^\top z_m(x))^2} \le \eta\}
    = \{x:
        z_m(x)^\top \Mmat^{-1} z_m(x)
    \le\eta\},
\end{align}
that is the $\eta$-sublevel set of the polynomial empirical inverse Christoffel
function.
Applying the PAC-Bayes theorem to this construction yields the following
alternative to Lemma~\ref{lem:sriskbound_kernel}.
\begin{lemma}
    \label{lem:sriskbound_poly}
   Let $x_1,\dotsc,x_N\simiid X$ denote a set of observations used to construct
   $C_{Q}$ from $C_{P}$
   in~\eqref{eq:poly_cfun_prior_family}. 
   The stochastic risk $\srisk$
   is bounded by $\overline{r}\in(0,1)$, 
   where
   \begin{equation}
       \label{eq:pacbayes_poly}
       \overline{r} = \sup \left\{\hspace{-2pt} \beta : 
       D_{\text{ber}}(\esrisk || \beta)
       \hspace{-1pt}\le\hspace{-1pt}
       \frac{D_{KL}(\Norm{0,(\sigma_0^{2}I \hspace{-1pt}+\hspace{-1pt} \Mmat)^{-1}}\hspace{-2pt}|| \Norm{0,\sigma_0^{-2}I})
           \hspace{-1pt}+\hspace{-1pt} \log\frac{N+1}{\delta}
   }{N}\hspace{-2pt}\right\}.
   \end{equation}
\end{lemma}

Using this alternative lemma, we obtain a validation for
Algorithm~\ref{alg:cfun_pacbayes_poly}.
\begin{corollary}
    \label{cor:alg_cfun_pacbayes_poly}
    At each stage $i$ of execution, the empirical inverse Christoffel function
    constructed in Algorithms~\ref{alg:cfun_pacbayes_poly} satisfies the PAC
    bound~\eqref{eq:pac_epsi}.
\end{corollary}
\begin{proof}
    The argument to verify Algorithm~\ref{alg:cfun_classical} is identical to
    that used in the proof of Theorem~\ref{thm:alg_cfun_pacbayes}, except that
    Lemma~\ref{lem:sriskbound_poly} is used instead of
    Lemma~\ref{lem:sriskbound_kernel}.
\end{proof}

\begin{remark}
    \label{rmk:thresholds}
    Algorithms~\ref{alg:cfun_pacbayes} and~\ref{alg:cfun_pacbayes_poly} require
    that a threshold parameter $\eta$ be selected \emph{a priori} based on the
    kernel.
    For instance, if a squared exponential kernel
    $k(x,y)=\exp(-\|x-y\|^2/(2\ell)^2)$ is used in
    Algorithm~\ref{alg:cfun_pacbayes}, the resulting empirical inverse
    Christoffel function will always have values in $[0,1]$, with values
    generally smaller close to data points: thus choosing a value between $0$
    and $1$ is a suitable choice, with smaller values yielding finer
    approximations of the support set. For
    Algorithm~\ref{alg:cfun_pacbayes_poly}, a reasonable heuristic is to select
    $\eta=\binom{n+2m}{n}/\epsilon$: one can show that the expected
    value of the true inverse Christoffel function of order $m$ is
    $\binom{n+2m}{n}$ when the input is distributed according to $X$, so by
    Markov's inequality the probability mass of the
    $\binom{n+2m}{n}/\epsilon$-level subset of the true inverse Christoffel
    function is at least $1-\epsilon$.
\end{remark}

\subsection{Numerical Considerations for Large Datasets}%
\label{sub:numerical_considerations_for_large_datasets}

As the sample size $N$ grows, the calculations in
Algorithm~\ref{alg:cfun_pacbayes} involving the kernel matrix $K$ can become
computation- and memory-intensive. 
In particular,
evaluating $\kappa^{-1}(x)$
to compute the support set estimate and 
computing the KL divergence
that appears in~\eqref{eq:pacbayes_kernel}
both require the construction of an $N\times N$ matrix
and an $O(N^3)$ matrix inversion.
Computational difficulties related to the size of the $K$ matrix
are well known in the field of kernel machines; in response,
a wealth of approximation techniques have been developed to reduce compute and
memory requirements at the cost of fidelity. These approximation techniques can
be used to improve the efficiency of evaluating the kernelized empirical inverse
Christoffel function and its construction via Algorithm~\ref{alg:cfun_pacbayes}.

\newcommand{\Knys}{\tilde{K}}
\newcommand{\Nnys}{r}
For example, to reduce the speed and memory requirements of evaluating
$\kappa^{-1}(x)$, we can replace the kernel matrix $K$ with its rank-$\Nnys$
Nystr\"om approximation~\cite{williams2001using}.
The Nystr\"om approximation is a method to construct low-rank
approximations of Gramian matrices, such as the kernel matrix $K$,
which has a simple expression in terms of
block submatrices of the original matrix. Specifically,
the rank-$\Nnys$ Nystr\"om approximation of the kernel matrix $K$ has the form
\begin{equation}
    \label{eq:nys_def}
    \Knys = K_{N\Nnys}K_{\Nnys\Nnys}^{-1}K_{N\Nnys},
\end{equation}
where $K_{N\Nnys}\in\R^{N\times \Nnys}$, $K_{\Nnys\Nnys}\in\R^{\Nnys\times \Nnys}$ are submatrices of
$K$ whose $i,j$ elements are $k(x_i,x_j)$. 
Making the substitution $K\mapsto\Knys$ and applying the matrix
inversion lemma to $\kappa^{-1}(x)$ yields
\begin{equation}
    \begin{aligned}
        \tilde{\kappa}^{-1}(x)
        &= k(x,x) - \sigma_0^{-2}\bigg(k_D(x)^\top k_D(x)\\
        &\quad - 
        (K_{N\Nnys}k_D(x))^\top
        (I - K_{N\Nnys }\left(\sigma_0^2 K_{\Nnys \Nnys } + K_{\Nnys N}K_{N\Nnys }\right)^{-1}
    (K_{\Nnys N}k_D(x))\bigg).
    \end{aligned}
\end{equation}
To numerically compute the final expression, 
we need only invert an $\Nnys\times\Nnys$ matrix instead of an $N\times N$ one; indeed,
we do not need to explicitly
construct an $N\times N$ matrix at all.

Next, we consider a method to over-approximate the KL divergence based on the
$\Nnys$ largest eigenvalues of $K$. Since the 
KL divergence $D_{KL}(Z_0||Z_1)$ between $N$-dimensional normal random variables
$Z_0\sim\Norm{\mu_0,\Sigma_0}$ and
$Z_1\sim\Norm{\mu_1,\Sigma_1}$
has the expression
\begin{equation}
    \label{eq:normal_kl_generic} 
    D_{KL}(Z_0||Z_1)
    =
    \tfrac{1}{2}\log\det 
    \Sigma_1\Sigma_0^{-1}
    +
    \tfrac{1}{2}\tr
    \Sigma_1^{-1}\left(
        (\mu_0-\mu_1)(\mu_0-\mu_1)^\top + \Sigma_0
    \right)-\tfrac{N}{2}.
\end{equation}
For $\Sigma_0=(\sigma_0^{-2}I + K^{-1})^{-1}$, $\Sigma_1=K$,
$\mu_0=\mu_1=0$, \eqref{eq:normal_kl_generic} reduces to
\begin{align}
        \tfrac{1}{2}\log\det(I+\sigma_0^{-2} K)
        +
        \tfrac{1}{2}\tr\left((I+\sigma_0^{-2}K)^{-1}\right)
        -\tfrac{N}{2}.
        \label{eq:kldiv_derivation_1}
\end{align} 
Since $\log(1+\sigma_0^{-2}x)$ and $1/(1+\sigma_0^{-2}x)$ are analytic for $x
\ge 0$, we can apply the spectral mapping theorem~\cite[Sec. 4.7]{callier2012linear}
to~\eqref{eq:kldiv_derivation_1} to obtain an expression for the KL divergence
in terms of the 
eigenvalues $\lambda_1,\dotsc,\lambda_N$ of $K$, namely
\begin{align}
        \quad=
        \frac{1}{2}
        \sum_{i=1}^N
        \left(
        \log(1+\sigma_0^{-2}\lambda_i) +
    \frac{1}{1+\sigma_0^{-2}\lambda_i}-1\right).
        \label{eq:kldiv_derivation_2}
\end{align}
Numerically computing the KL divergence with the
expression~\ref{eq:kldiv_derivation_1} requires an explicit construction of the
$K$ matrix, and the inverse of an $N\times N$ matrix: this requires $O(N^3)$
operations and $O(N^2)$ memory.
Using~\eqref{eq:kldiv_derivation_2} instead of~\eqref{eq:kldiv_derivation_1} to
compute the KL divergence with the full set of eigenvalues does not generally
yield an improvement, since computing the eigenvalues of $K$ is also $O(N^3)$.
However, since $K$ is a symmetric positive definite matrix, the eigenvalues are
all positive, and the $m$ largest eigenvalues can be computed in less than
$O(N^3)$ time, for instance by a Lanczos-type algorithm~\cite[ch. 9]{van1996matrix}.
Let $\lambda_p$ denote the $p^{th}$ largest eigenvalue:
Since~\eqref{eq:kldiv_derivation_2} is a nondecreasing function in each 
$\lambda_i$, the approximation $\lambda_i \approx \lambda_p$ for $\lambda_i$ such that $\lambda_i < \lambda_p$ yields an
upper bound on the KL divergence that can be computed in less than $O(N^3)$
time. 

\section{Examples}%
\label{sec:examples}

This section demonstrates how
Algorithms~\ref{alg:cfun_classical},~\ref{alg:cfun_pacbayes},
and~\ref{alg:cfun_pacbayes_poly}
can be used to make
accurate estimates of forward reachable sets. 
These examples were run on Savio, a high-performance computing cluster managed
by the University of California at Berkeley. Specifically, each experiment used
a single \texttt{savio2\_bigmem} node comprising 20
CPUs running at 2.3 GHz and 128 GB of memory.
In all experiments, we use the parameters $\epsilon=0.1$, $\delta=10^{-9}$ for
all three algorithms, and in Algorithm~\ref{alg:cfun_pacbayes}, we use the squared
exponential kernel $k(x,y)=\exp(-\|x-y\|^2/(2\ell)^2)$. The values for $m$
and $\ell$ used in experiments is listed in Table~\ref{tab:times}.
To select thresholds in Algorithms~\ref{alg:cfun_pacbayes}
and~\ref{alg:cfun_pacbayes_poly}, we follow the advice of
Remark~\ref{rmk:thresholds}, using $\eta=0.15$ for
Algorithm~\ref{alg:cfun_pacbayes} and $\eta=\binom{n+2m}{n}/\epsilon$ for
Algorithm~\ref{alg:cfun_pacbayes_poly}.
For Algorithm~\ref{alg:cfun_classical}, we use an initial sample size of $20,000$ and 
a batch size of $5,000$ samples. For Algorithm~\ref{alg:cfun_pacbayes_poly}, we
use an initial sample size and batch size of $1,000$ samples. 

\begin{table}[]
\centering
\begin{tabular}{@{}lcclcclccl@{}}
\toprule
\multirow{2}{*}{Example} 
& \multicolumn{3}{c}{Alg.~\ref{alg:cfun_classical}} 
& \multicolumn{3}{c}{Alg.~\ref{alg:cfun_pacbayes_poly}} 
& \multicolumn{3}{c}{Alg.~\ref{alg:cfun_pacbayes}} \\
\cmidrule(lr){2-4}
\cmidrule(lr){5-7}
\cmidrule(lr){8-10}
& $m$ & time (s) & $N$ & $m$ & time (s) & $N$ & $\ell$ & time (s) & $N$   \\ \midrule
Duffing & 10 & 39  & 70307 & 10 & 13 & 11000 & $1/4$ & 506 & 30000 \\
Quadrotor & 4 & 3  & 14587 & 4 & 4 & 6000 & $1/4$ & 488  & 35000 \\
Traffic &10 & 16 & 70307 & 10 & 11 & 10000 & $1/4$ & 325  & 30000  \\ \bottomrule
\end{tabular}
\caption{Computation times, sample sizes, and Christoffel function parameters for numerical experiments. All times in seconds.
    Algorithms~\ref{alg:cfun_classical} and~\ref{alg:cfun_pacbayes_poly} used
    polynomial order $m$, and Algorithm~\ref{alg:cfun_pacbayes} used
    $k(x,y)=\exp(-\|x-y\|^2/(2\ell)^2)$, with $m$, $\ell$ as given in the
table. All experiments use $\epsilon=0.1$, $\delta=10^{-9}$.}
\label{tab:times}
\end{table}

\subsection{Chaotic Nonlinear Oscillator}%
\label{subsec:chaotic_nonlinear_oscillator}

\begin{wrapfigure}{r}{0.55\textwidth}
    \centering
    \includegraphics[width=\linewidth]{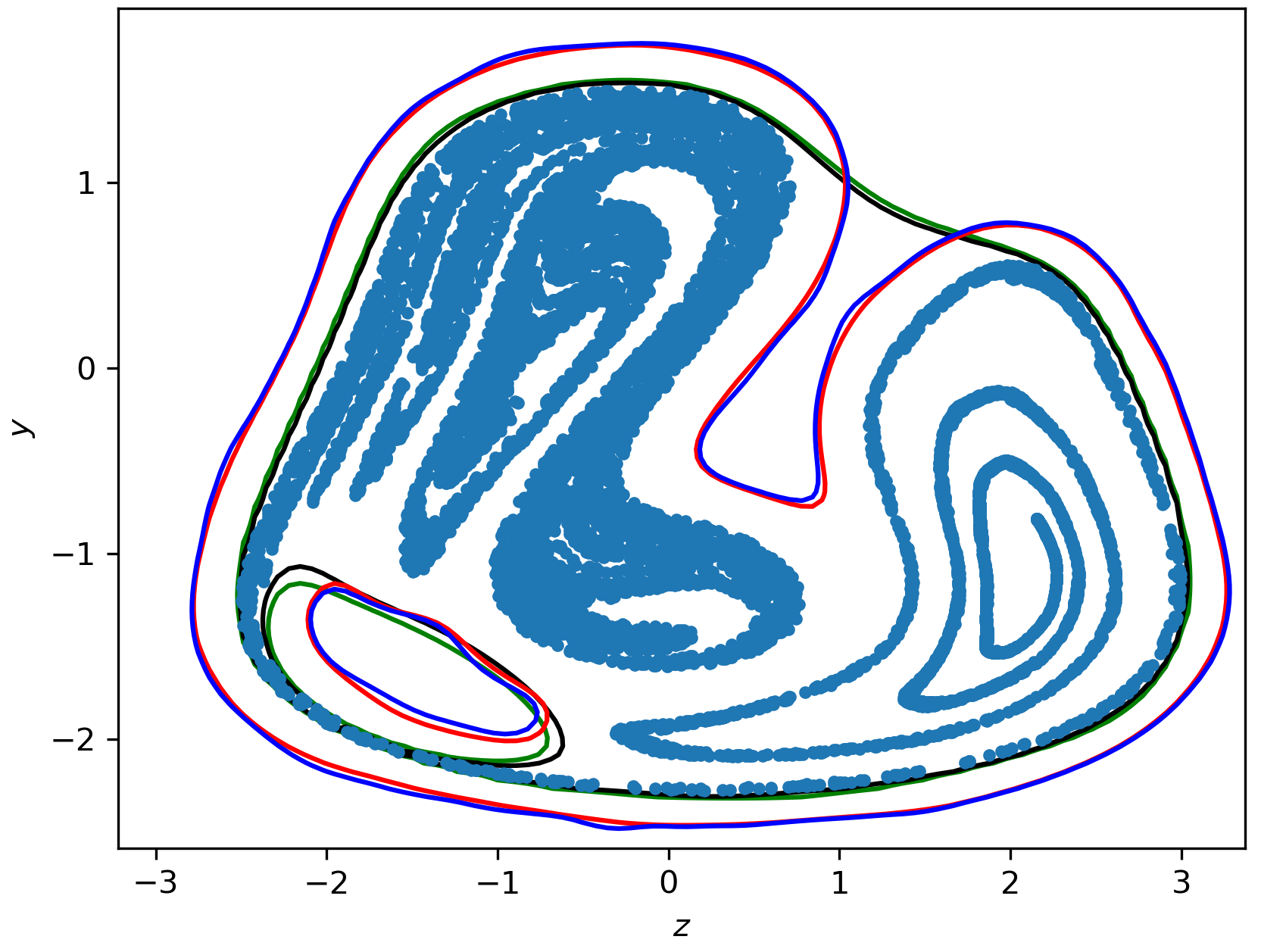}
    \caption{Results of Algorithms~\ref{alg:cfun_classical}, \ref{alg:cfun_pacbayes}
        and~\ref{alg:cfun_pacbayes_poly} on the 
        Duffing oscillator reachability problem. 
    Black contour: output of Algorithm~\ref{alg:cfun_classical}.
    Green contour: output of Algorithm~\ref{alg:cfun_pacbayes_poly}.
    Red contour: output of Algorithm~\ref{alg:cfun_pacbayes}.
    Blue contour: output of Algorithm~\ref{alg:cfun_pacbayes}, over-approximated
    using the Nystr\"om approximation with 1,000 samples.
    Blue dots: samples used in Algorithm~\ref{alg:cfun_pacbayes_poly}.
}%
    \label{fig:duffing}
\end{wrapfigure}

The first example is a reachable set estimation problem for the nonlinear,
time-varying
system with dynamics
$
\dot{z} =  y,
\dot{y} = -\alpha y + z - z^3 + \gamma\cos(\omega t),
$
with states $x=(z,y)\in\R^2$ and parameters $\alpha, \gamma, \omega\in\R$. This system
is known as the \emph{Duffing oscillator}, a nonlinear oscillator which exhibits
chaotic behavior for certain values of $\alpha$, $\gamma$, and $\omega$, for instance
$\alpha = 0.05$, $\gamma = 0.4$, $\omega = 1.3$.
The initial set is the interval such that $z(0)\in[0.95, 1.05]$,
$y(0)\in[-0.05,0.05]$, and we take $X_0$ to be uniform over this interval. The time range is $[t_0,t_1]=[0,100]$.

We use Algorithms~\ref{alg:cfun_classical}
and~\ref{alg:cfun_pacbayes_poly}
to compute  reachable set estimates using an order $k=10$
empirical inverse Christoffel function with accuracy and confidence parameters
$\epsilon=0.10$, $\delta=10^{-9}$.
Additionally, we use Algorithm~\ref{alg:cfun_pacbayes} to compute a kernelized
empirical inverse Christoffel function using the squared exponential kernel
$k(x,y)=\exp(\|x-y\|^2/(2\ell^2))$ with $\ell=0.25$.
Figure~\ref{fig:duffing} shows the reachable set estimate
for the Duffing oscillator system with the problem data given above produced by
all three algorithms: for Algorithm~\ref{alg:cfun_pacbayes}, both the full
kernelized Christoffel function estimator and its Nystr\"om approximation with
$r=2000$.
The cloud of points are the $11,000$ samples used in Algorithm~\ref{alg:cfun_pacbayes_poly}. 
The reachable set
estimate is neither convex nor simply connected, closely following the
boundaries of the cloud of points and excluding an empty region. In particular,
all estimates exhibit a hole in a region of the state space devoid of samples. 

\subsection{Planar Quadrotor}

The next example is a reachable set estimation problem for horizontal position
and altitude in a nonlinear model of the planar dynamics of a quadrotor used as
an example  in~\cite{mitchell2019invariant,bouffard2012board}. 
The dynamics for this model are
$
        \ddot{p_x} = u_1 K\sin(\theta),
        \ddot{p_h} = -g + u_1 L\cos(\theta),
        \ddot{\theta} = -d_0\theta - d_1\dot{\theta} + n_0 u_2, 
$
where $p_x$ and $p_h$ denote the quadrotor's horizontal position and altitude in
meters,
respectively, and $\theta$ denotes its angular displacement (so that the
quadrotor is level with the ground at $\theta=0$) in radians. 
The system has 6
states, which we take to be $x$, $h$, $\theta$, and their first derivatives. The
two system inputs $u_1$ and $u_2$ (treated as disturbances for this example) represent the motor thrust and the desired angle,
respectively. The parameter values used (following~\cite{bouffard2012board}) are $g=9.81$,
$L=0.64$, $d_0=70$, $d_1=17$, and $n_0=55$. The set of initial states is the interval such that 
        $p_x(0)\in[-1.7, 1.7]$,  $\dot{p}_x(0)\in[-0.8, 0.8]$, 
        $p_h(0)\in[0.3, 2.0]$,  $\dot{p}_h(0)\in[-1.0, 1.0]$, 
        $\theta(0)\in[-\pi/12, \pi/12]$,  $\dot{\theta}(0)\in[-\pi/2, \pi/2]$,
the set of inputs is the set of constant functions $u_1(t)=u_1$, $u_2(t)=u_2$ $\forall t\in[t_0,t_1]$, whose values lie in the
interval
        $u_1\in[-1.5+ g/L, 1.5 + g/L],  u_2\in[-\pi/4, \pi/4]$,
and we take $X_0$ and $D$ to be the uniform random variables defined over these intervals.
The time range is $[t_0,t_1]=[0,5]$. We take probabilistic parameters $\epsilon=0.10$, $\delta=10^{-9}$.
Since the goal of this example is to estimate a reachable set for the horizontal
position and altitude only, we are interested in a reachable set for a subset of
the state variables, namely $p_x$ and $p_h$.
Following Remark~\ref{rmk:rs_subset},
we use the reduced-state variations of
Algorithms~\ref{alg:cfun_classical},~\ref{alg:cfun_pacbayes},
to compute reachable set estimates using only data for the $(p_x,p_h)$ states,
effectively reducing the dimension of the problem from 6 to 2.
Figure~\ref{fig:quadrotor} shows the reachable set estimate
for the planar quadrotor system with the problem data given above produced by
all three algorithms and the Nystr\"om-approximated
Algorithm~\ref{alg:cfun_pacbayes} with $r=2000$.
The reachable set estimates displayed in Figure~\ref{fig:quadrotor}, and the
computation times reported in Table~\ref{tab:times}, use the reduced-state
variation.

\begin{wrapfigure}{r}{0.55\textwidth}
    \centering
    \includegraphics[width=\linewidth]{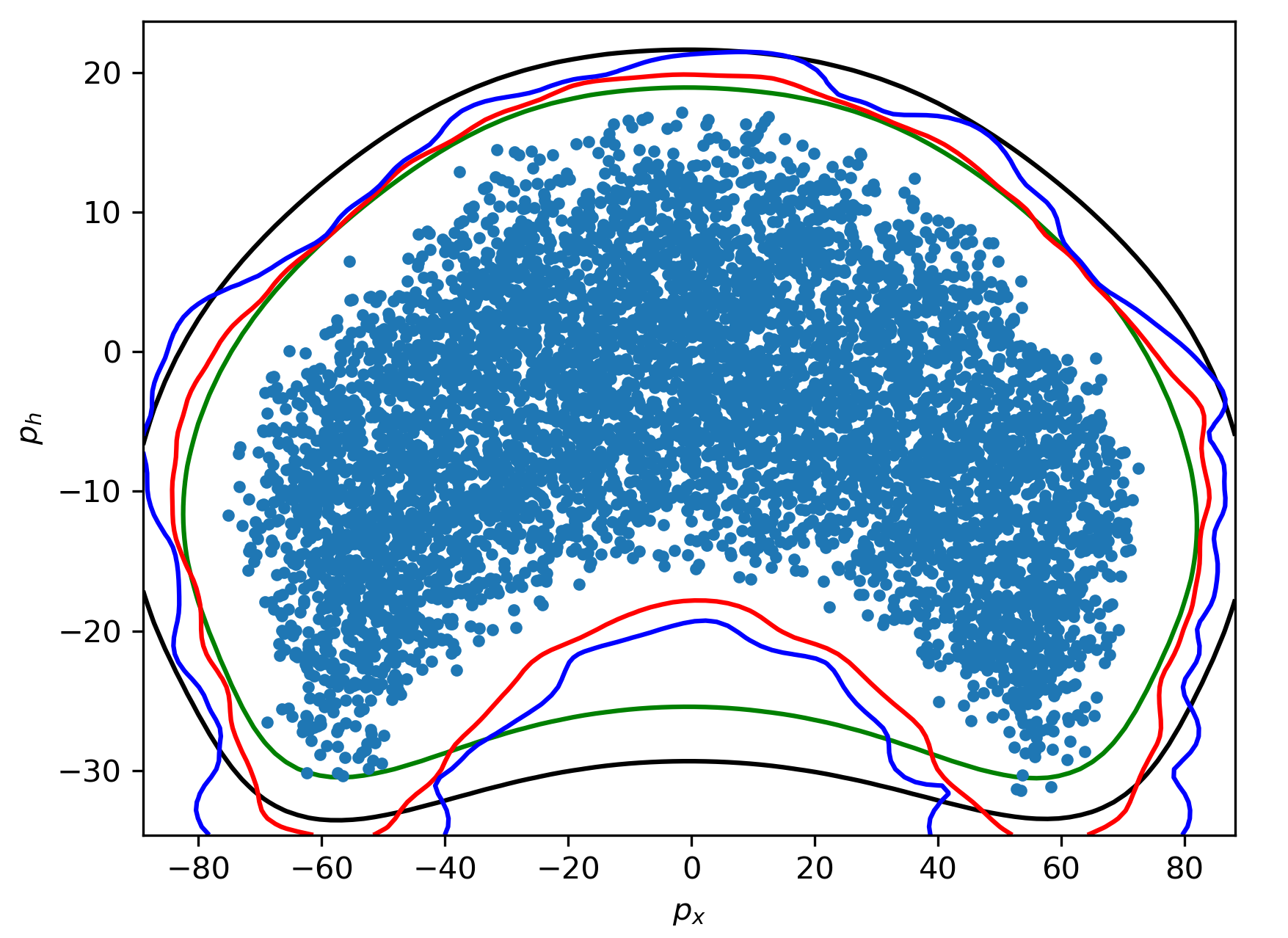}
    \caption{Results of Algorithms~\ref{alg:cfun_classical}, \ref{alg:cfun_pacbayes},
        and~\ref{alg:cfun_pacbayes_poly} on the 
        planar quadrotor reachability problem, restricting the reachability
        analysis to the $(p_x,p_h)$ plane.
    Green contour: polynomial Christoffel
    function of order $k=10$. Blue contour: kernelized inverse Christoffel function
    with squared exponential kernel. Red contour: Nystr\"om approximation
    ($m=10,000$) of the kernelized inverse Christoffel function with squared
    exponential kernel.}%
    \label{fig:quadrotor}
\end{wrapfigure}

\subsection{Monotone Traffic}

This example is a special case of a continuous-time road traffic analysis problem
used as a reachability benchmark in~\cite{coogan2018benchmark}. This problem
investigates the density of traffic on a single lane over a time range over four 
periods of duration $T$ using 
the Cell Transmission Model~\cite{daganzo1994cell} that divides the road into $n$
equal segments. The spatially discretized model
is an $n$-dimensional dynamical system with states $x_1,\dotsc,x_n$, where $x_i$ represents the
density of traffic in the $i^{th}$ segment.
Traffic enters segment through $x_1$ and flows through each successive segment before leaving through
segment $n$.
The system dynamics~\eqref{eq:traffic_dynamics} are monotone, i.e.
order-preserving: this property allows us to compute an interval containing the reachable set by evaluating the dynamics at the extreme points of the intervals defining the initial set and the set of disturbances.
While this interval
over-approximation is easy to compute, and is the best possible over-approximation
by an interval, it is in general a conservative over-approximation because the
reachable set may only occupy a small volume of the interval. Since the
empirical Inverse Christoffel function method can accurately detect the geometry
of the reachable set, we use this method to compare the shape of the reachable
set to the best interval over-approximation. 

The state dynamics are
\begin{equation}
\begin{split}
    \label{eq:traffic_dynamics}
    \dot{x}_1 &= \frac{1}{T}\left(d-\min(c, vx_{1}, w(\overline{x}-x_{2}))\right)\\
    \dot{x}_i &= \frac{1}{T}\big( \min(c, vx_{i-1}, w(\overline{x}-x_{i})) \\
              &- \min(c, vx_{i}, w(\overline{x}-x_{i+1}))\big), \quad(i=2,\dotsc,n-1)\\
    \dot{x}_{n} &= \frac{1}{T}\left(\min(c, vx_{n-1}, w(\overline{x}-x_{n})/\beta) - \min(c, vx_{n}))\right),
\end{split}
\end{equation}
where $v$ represents the free-flow speed of
traffic, $c$ the maximum flow between neighboring
segments, $\bar{x}$ the maximum occupancy of a segment, and $w$ the congestion
wave speed. The input $u$ represents the influx of traffic into the first node.
For the reachable set estimation problem, we use a model with $n=6$ states, and take $T=30$, $v=0.5$, $w=1/6$, and $\bar{x}=320$. The initial set is the
interval such that $x_i(0)\in[100,200]$, $i=1,\dotsc,n$, the set of
disturbances is the set of constant disturbances with values in the range $d\in[40/T, 60/T]$, and $X_0$ and $D$ are the uniform random variables over these sets. The time range is $[t_0, t_1]=[0, 4T]$.

\begin{wrapfigure}{r}{0.55\textwidth}
    \centering 
\includegraphics[width=\linewidth]{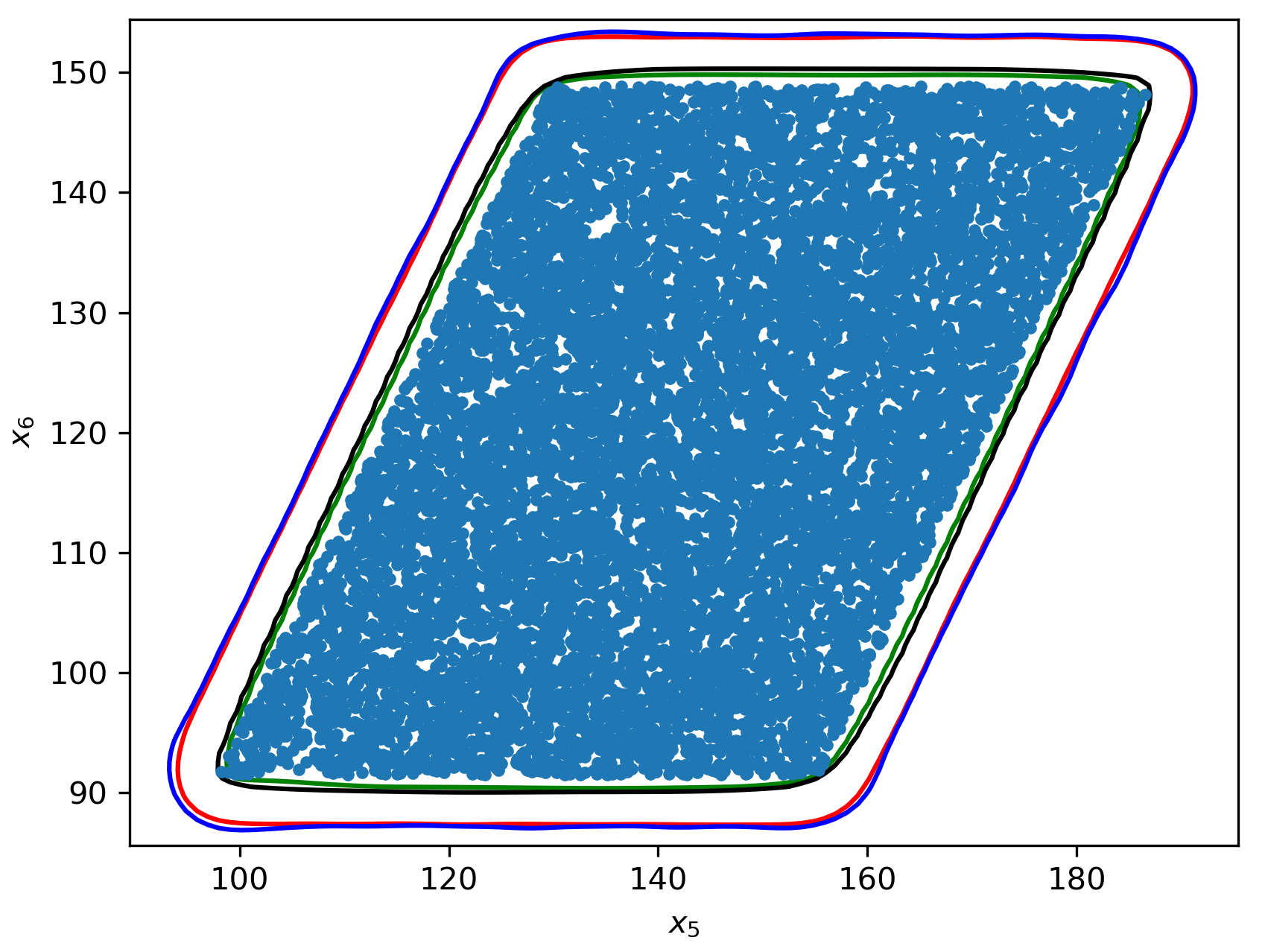}
\caption{Results of Algorithms~\ref{alg:cfun_classical},
    \ref{alg:cfun_pacbayes},
        and~\ref{alg:cfun_pacbayes_poly} on the 
        six-state monotone traffic reachability problem, restricting the reachability
        analysis to the $(x_5,x_6)$ plane.
    Green contour: polynomial Christoffel
    function of order $k=10$. Blue contour: kernelized inverse Christoffel function
    with squared exponential kernel. Red contour: Nystr\"om approximation
    ($m=10,000$) of the kernelized inverse Christoffel function with squared
    exponential kernel.}%
    \label{fig:traffic}
\end{wrapfigure}

We use the reduced-state variant of
Algorithms~\ref{alg:cfun_classical},~\ref{alg:cfun_pacbayes},
and~\ref{alg:cfun_pacbayes_poly}
to compute a reachable set for the traffic densities $x_5$ and $x_6$ at the
end of the road,
using an order $k=10$
empirical inverse Christoffel function with accuracy and confidence parameters
$\epsilon=0.10$, $\delta=10^{-9}$. 
Figure~\ref{fig:quadrotor} 
compares the reachable set estimates for the traffic system 
produced by all three algorithms, and the Nystr\"om-approximated
Algorithm~\ref{alg:cfun_pacbayes} with $r=2000$,
with the projection of the tight interval
over-approximation computed using the monotonicity property of the traffic
system.
The figure indicates that the tight interval over-approximation of the
reachable set is a somewhat conservative over-approximation, since the reachable
set has approximately the shape of a parallelotope whose sides are not axis-aligned.

\section{Conclusion}

This paper advances the non-asymptotic theory of support set estimation by
empirical Christoffel functions by applying the formal connection between
Christoffel functions and Gaussian process regression models to a Bayesian PAC
analysis of the estimator. 
The numerical examples demonstrate that the Bayesian
PAC give a large improvement in sample efficiency over classical PAC bounds.
Additionally, Bayesian PAC arguments endow the kernelized inverse Christoffel
function with PAC bounds, a development not possible
with classical VC dimension bounds.

Improvements to the general theory can advance in step with advances in Bayesian
PAC analysis. For instance, there are new results in theory of
\emph{derandomizing} Bayesian PAC bounds, which could offer sample efficiency
improvements over the argument used in
Lemma~\ref{lem:central_concept_risk_bound} to apply the Bayesian PAC bound to
the central concept. Furthermore, domain-specific knowledge could be applied to
the GP prior used to construct the Christoffel functions. For instance, in
reachability problems and estimate of the system sensitivity matrix could be
used to intelligently select length-scales in the kernel, along with other
algorithm hyper-parameters such as the initial sample size and batch size.

\bibliographystyle{siamplain}
\bibliography{refs}

\appendix

\section{Background on Gaussian Process Models}%
\label{sub:gaussian_process_models}

A Gaussian process $g$ is a stochastic process such that vectors
$(g(x_1),\dotsc,g(x_m))$ of point evaluations are multivariate
Gaussian distributions. Similar to how a Gaussian random variable is completely
characterized by its mean and variance, a Gaussian process is completely
characterized by a mean function $m$, defined pointwise as $m(x)=\Ex{g(x)}$, and
a positive semidefinite covariance function $k$, defined on all pairs of points
$x,y\in\dom$ as $k(x,y)=\Ex{g(x)g(y)}$.

Gaussian processes can also be defined according to a finite set of basis functions,
admitting a direct construction as a finite weighted sum.
For an
$m$-dimensional
space of functions with basis $b_1,\dotsc,b_m:\dom\to\R$,
we form the stochastic weighted
average $\sumio{m} w_i b_i$, where $w=(\liston[m]{w})\sim\Norm{0,\Sigma}$.
This weighted average is a Gaussian process whose support is the span of
$\liston[m]{b}$, with mean $m(x)=0$ and covariance $k(x,y)=\sumio{m}b(x)^\top
\Sigma b(y)$, where $b(\cdot) = (b_1(\cdot),\dotsc,b_m(\cdot))^\top$.

The Gaussian process regression model is Bayesian regression model that uses a
Gaussian process as the prior over regression functions.
In our case, we take the mean of the prior process to be zero.
The data is assumed to be of the form $g(x_i)=h_i+\varepsilon$, where
$\varepsilon$ is a Gaussian noise term with variance $\sigma^2$.
Under these conditions, the posterior for the unknown
function is also a Gaussian process, whose mean and covariance are given by the
formulas
\begin{align}
    \label{eq:gp_kernel}
    m_q(x) &= k_D(x)^\top{\left(\sigma^2 I_N + K\right)}^{-1}h,\\
    k_q(x,y) &= 
    k(x,y) - k_D(x)^\top{\left(\sigma^2 I_N + K\right)}^{-1} k_D(y).
\end{align}
In the finite-dimensional case,
the posterior process 
has mean and covariance functions
\begin{align}
    \label{eq:gp_feature}
    m_q(x) &= \sigma^{-2}b(x)^\top{\left(\Sigma^{-1} + \sigma^2 BB^\top\right)}^{-1}By \\
    k_q(x,y) &= {b(x)}^\top{(\Sigma^{-1} + \sigma^{-2}BB^\top)}^{-1}b(x),
\end{align}
where $B\in\R{m\times N}$ is the matrix formed by evaluating the basis
functions on the data, that is $B=[b(x_1)~~\cdots~~b(x_N)]$.
Taking $b=z_k$, $\Sigma=\sigma_0^{-2} I$, $\sigma = N^{-1/2}$, 
yields the posterior variance
$
    \Var[g_q]{x} = {z_m(x)}^\top{\left(\sigma_0^2I  + \frac{1}{N}\sumion z_m(x_i)z_m(x_i)\top\right)}^{-1}{z_m(x)},
    $
which is precisely the polynomial empirical inverse Christoffel function of
order $k$ for the data $x_1,\dotsc,x_n$ evaluated at the point $x$.

\section{Proofs of Some Results in
Section~\ref{sub:bayesian_pac_analysis_of_christoffel_function_estimators}}
\quad

\begin{proof}[Proof of Lemma~\ref{lem:empirical_stochastic_risk}]
    We consider the kernel case, since the polynomial case follows by the
    appropriate choice of kernel function.
    Recall that $\kappa^{-1}(x)$ is the variance of $g_p$ by construction.
    Evaluating $g_p$ at a single point $x$ yields the normal random variable
    $g_p(x)\sim\Norm{0, \kappa^{-1}(x)}$.
    It follows that $g_p(x)/\sqrt{\kappa^{-1}(x)}~\sim\Norm{0,1}$, and that
    $g_p(x)^2/\kappa^{-1}(x)\sim\chi^2_1$, that is that
    $g_p(x)^2/\kappa^{-1}(x)$, is a chi-square random variable with one degree
    of freedom. The average loss over $C_P$ for a fixed point $x$ is then
    \begin{equation}
       \begin{aligned}
           \label{eq:a}
           \Ex{\ell(C_Q,x)}&=\Ex{\ind{x\in C_Q}}\\
                           &= 1-\Pr{g_p(x)^2 \le \eta} 
                            = 1-\Pr{\frac{g_p(x)^2}{\kappa^{-1}(x)} \le \frac{\eta}{\kappa^{-1}(x)}}\\
                           &= 1 - F_1\left(\frac{\eta}{\kappa^{-1}(x)}\right).
       \end{aligned} 
    \end{equation}
    Averaging this expression over the data points yields~\eqref{eq:poly_empirical_stochastic_risk}.
\end{proof}

\begin{proof}[Proof of Lemma~\ref{lem:sriskbound_poly}]
    We apply the Seeger PAC-Bayes Theorem~\ref{thm:pacbayes} to the prior and posterior
    measures $P$ and $Q$ induced by $C_{P}$ and $C_{Q}$ as
    defined in~\eqref{eq:poly_cfun_prior_family}.
    Recall that these prior and posterior measures are defined by the random
    vectors
    $W_P\sim\Norm{0, \sigma_0^{-2}I}$,
    $W_Q\sim\Norm{0, (\sigma_0^{-2}I + \Mmat)^{-1}}$,
    which act as parameters.
    Applying this choice of $W_p$ and $W_q$ to equation~\eqref{eq:pacbayes}
    of Theorem~\ref{thm:pacbayes} yields the inequality
   \begin{equation}
   P_X^N\left(\left\{ x_1,\dotsc,x_N : 
   D_{\text{ber}}(\esrisk || \srisk)
   \le \gamma\right\}\right) \ge 1-\delta, 
   \end{equation}
   where 
   $\gamma=\tfrac{1}{N}(D_{KL}(\Norm{0,(\sigma_0^{2}I + \Mmat)^{-1}} || \Norm{0,\sigma_0^{-2}I}) + \log\frac{N+1}{\delta})$.
    Suppose the data set $x_1,\dotsc,x_N$ is one such that the inner inequality
    $ D_{\text{ber}}(\esrisk || \srisk) \le\gamma $ 
    holds: then $\srisk$, the true stochastic risk, lies in the set
    $\{\beta : D_{ber}(\esrisk||\beta) \le\gamma\}$.
    The function $D_{ber}(\esrisk || \beta)$ is convex in $\beta$ and covers
    the range $[0,\infty)$, attaining $0$ for $\beta=\esrisk$ and approaching
    $\infty$ for $\beta\to 0$ and $\beta\to 1$.
    By these properties,
    $\{\beta : D_{ber}(\esrisk||\beta) \le\gamma\}$
    is a closed convex subset of $(0,1)$ for any positive $\gamma$.
    As such, it attains a supremum,
    meaning that $\overline{r}$ as defined in~\eqref{eq:pacbayes_poly} is
    well-defined. Thus we have, with confidence $1-\delta$, that $\overline{r}$
    is an upper bound on the stochastic risk $\srisk$.
\end{proof}

\begin{proof}[Proof of Lemma~\ref{lem:sriskbound_kernel}]
    As in the proof of Lemma~\ref{lem:sriskbound_poly}
    we apply the Seeger PAC-Bayes Theorem~\ref{thm:pacbayes}, 
    this time to the prior and posterior
    measures $P$ and $Q$ induced by $C_{P}$ and $C_{Q}$ as
    defined in~\eqref{eq:kernel_prior_posterior_families}.
    These measures are defined by the Gaussian processes
    $g_p$ and $g_q$ which act as the concept class parameters
    $W_P$ and $W_Q$ respectively in the statement of Theorem~\ref{thm:pacbayes}. 
    To compute the KL divergence between $W_P$ and $W_Q$, we use another result
    due to Seeger, described in Section 2.2 of~\cite{seeger2002pac}, which states that
    the KL divergence between a prior Gaussian process $g_p$ and the posterior
    Gaussian processes $g_q$ obtained after conditioning on data
    $x_1,\dotsc,x_N$ is equal to the KL divergence between the restriction of
    the two Gaussian processes to the data points, that is the KL divergence
    between the multivariate normal random vectors
    $(g_p(x_1),\dotsc,g_p(x_N))$
    and
    $(g_q(x_1),\dotsc,g_q(x_N))$.
    The mean and covariance of these random variables are simply the restrictions
    of the mean and covariance functions of their defining processes to
    $(x_1,\dotsc,x_N)$. Both random vectors have mean zero. The covariance
    matrix of the prior random vector
    $(g_p(x_1),\dotsc,g_p(x_N))$ is $K_p(X,X)=K(X,X)$ as discussed in
    Section~\ref{sub:gaussian_process_models}. By~\eqref{eq:gp_kernel} and an application of the matrix inversion lemma,
    the covariance of the posterior random vector
    $(g_q(x_1),\dotsc,g_q(x_N))$ is
    \begin{equation}
        \begin{aligned}
            K_q(X,X) 
            &= 
            K(X,X) - K(X,X)\left(\sigma_0^2 I + K(X,X)\right)^{-1}K(X,X)\\
            &=
            \left(K(X,X)^{-1} + \sigma_0^{-2} I\right)^{-1}.
        \end{aligned}
    \end{equation} 
\end{proof}

\begin{proof}[Proof of Lemma~\ref{lem:central_concept_risk_bound}]
    Consider a point $x\in\dom$ outside of the central concept, that is such
    that $\bar{c}_\eta(x) = \Ex{(g(x)^2} > \eta$. The probability that $W_Q^\top z_m(x)$ also
    exceeds $\eta$ is bounded as
    \begin{align}
        \Pr{(g(x)^2 \ge \eta}
        &\le
        \Pr{(g(x)^2 \ge \Ex{(g(x)^2}}
        =\Pr{\frac{(g(x)^2}{\Ex{(g(x)^2}} > 1}
        =1-F_1(1).
    \end{align}
    Next, let us consider the risk of the stochastic estimator, that is
    $\srisk=\Pr{(g(X)^2 > \eta}$. Applying the law of total probability
    with respect to the random variable $X$, we divide $\srisk$ into two
    integrals according to whether the central concept exceeds $\eta$:
    \begin{align}
        \Pr{(g(X)^2 > \eta}
        &=
        \int_\dom \Pr{(g(x)^2 > \eta} dP_x(x)\\
        &=
        \int_\dom
        \Pr{(g(x)^2 > \eta}
        \ind{\Ex{(g(x)^2} > \eta}
        dP_x(x)\\
        &+
        \int_\dom
        \Pr{(g(x)^2 > \eta}
        \ind{\Ex{(g(x)^2} \le \eta}
        dP_x(x).
    \end{align}
    We have that
    $
        \Pr{(g(X)^2 > \eta}
        \ge
        \int_\dom
        \Pr{(g(x)^2 > \eta}
        \ind{\Ex{(g(x)^2}~>~\eta}
        dP_x(x),
        $
    since all three integrands are nonnegative.
    To find an upper bound on this
    probability in terms of the empirical classifier, we combine the two
    inequalities above to find
    \begin{align}
        \Pr{(g(X)^2 > \eta}
        &=
        \int_\dom \Pr{(g(x)^2 > \eta} dP_x(x)\\
        &\ge
        \int_\dom \Pr{(g(x)^2 > \eta}
             \ind{\Ex{(g(x)^2} > \eta} dP_x(x)\\
        &\ge
        (1-F_1(1))
        \int_\dom \ind{\Ex{(g(x)^2} > \eta} dP_x(x)\\
        &= (1-F_1(1))\Pr{\Ex{(g(x)^2} > \eta}
        = (1-F_1(1)) r(\hat{c}_\eta),
    \end{align}
    which we rearrange to yield
    $
        r(\bar{c}_\eta) \le \frac{1}{1-F_1(1)}\srisk[Q_\eta].
    $
\end{proof}

\end{document}